\documentclass[runningheads]{llncs}

\usepackage[T1]{fontenc}
\usepackage{graphicx}

\PassOptionsToPackage{hyphens}{url}

\usepackage{url}
\usepackage{amsfonts}
\usepackage{amssymb}

\usepackage{amsthm}

\usepackage{bbding}

\usepackage{todonotes}
\usepackage{amsmath}
\usepackage{subcaption}
\usepackage{framed}
\usepackage{booktabs}
\usepackage{mathtools}
\usepackage{enumitem}

\usepackage{pgfplots}
\pgfplotsset{compat=1.18}
\usepackage{nicefrac}

\usepackage{xspace}
\usepackage{multirow}
\usepackage{framed}

\usepackage{color}
\definecolor{deepblue}{rgb}{0,0,0.5}
\definecolor{deepred}{rgb}{0.6,0,0}
\definecolor{deepgreen}{rgb}{0,0.5,0}

\usepackage{listings}

\DeclareFixedFont{\ttb}{T1}{txtt}{bx}{n}{8} 
\DeclareFixedFont{\ttm}{T1}{txtt}{m}{n}{8}  

\usepackage{tikz}
\usetikzlibrary{matrix}
\usetikzlibrary{arrows.meta}
\usetikzlibrary{decorations.markings}

\usepackage[misc]{ifsym}

\usepackage{cite}
\usepackage{hyperref}
\hypersetup{
  colorlinks = true,
  citecolor = blue,
  linkcolor = blue,
  urlcolor = blue
}

\usepackage{color}

\usepackage{cleveref}

\newcommand{\kw}[1]{\ensuremath{\mathsf{#1}}}
\newcommand{\plus}{\kw{true}}
\newcommand{\minus}{\kw{false}}

\usepackage{array}
\newcommand{\PreserveBackslash}[1]{\let\temp=\\#1\let\\=\temp}
\newcolumntype{C}[1]{>{\PreserveBackslash\centering}p{#1}}
\newcolumntype{R}[1]{>{\PreserveBackslash\raggedleft}p{#1}}
\newcolumntype{L}[1]{>{\PreserveBackslash\raggedright}p{#1}}

\newcommand{\orientp}[3]{\sigma(#1,#2,#3)}
\newcommand{\orientn}[3]{\lnot{\sigma(#1,#2,#3)}}
\newcommand{\relaxp}[5]{r(#1,#2,#3,#4,#5)}
\newcommand{\relaxn}[5]{\lnot{r(#1,#2,#3,#4,#5)}}

\newcommand{\evalmaxsat}{\textsf{EvalMaxSAT}\xspace}
\newcommand{\maxcdcl}{\textsf{MaxCDCL}\xspace}
\newcommand{\pacose}{\textsf{Pacose}\xspace}
\newcommand{\maxhs}{\textsf{MaxHS}\xspace}
\newcommand{\veripb}{\textsf{VeriPB}\xspace}

\newcommand{\veritas}{\textsf{VeritasPBLib}\xspace}

\makeatletter
\def\orcidID#1{\href{http://orcid.org/#1}{\protect\raisebox{-1.25pt}{\protect\includegraphics{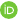}}}}
\makeatother

\begin{document}
%
\title{Automated Mathematical Discovery and Verification: Minimizing Pentagons in the Plane}
\author{Bernardo Subercaseaux {\Letter} \orcidID{0000-0003-2295-1299}  \and John Mackey \orcidID{0000-0001-7319-4377} \and \\ Marijn J.H. Heule \orcidID{0000-0002-5587-8801} \and Ruben Martins \orcidID{0000-0003-1525-1382}}
\authorrunning{Subercaseaux et al.}
\titlerunning{Automated Reasoning for Minimizing Pentagons in the Plane}
\institute{Carnegie Mellon University, Pittsburgh, PA 15213, USA\\\email{\{bsuberca, jmackey, mheule, rubenm\}@andrew.cmu.edu}}
%
%
\maketitle

\begin{abstract}
  We present a comprehensive demonstration of how automated reasoning can assist mathematical research, both in the discovery of conjectures and in their verification. 
Our focus is a discrete geometry problem: \emph{What is $\mu_{5}(n)$, the minimum number of convex pentagons induced by $n$ points in the plane?}
In the first stage toward tackling this problem, automated reasoning tools guide discovery and conjectures: we use SAT-based tools to find abstract configurations of points that would induce few pentagons. Afterward, we use Operations Research tools to find and visualize realizations of these configurations in the plane, if they exist.  
Mathematical thought and intuition are still vital parts of the process for turning the obtained visualizations into general
constructions. A surprisingly simple upper bound follows from our constructions: $\mu_{5}(n) \leq \binom{\lfloor n/2 \rfloor}{5} + \binom{\lceil n/2 \rceil}{5}$, and we conjecture it is optimal. In the second stage, we turn our focus to verifying this conjecture. Using MaxSAT, we confirm that $\mu_5(n)$ matches the conjectured values for $n \leq 16$, thereby improving both the existing lower and upper bounds for $n \in [12, 16]$. Our MaxSAT results rely on two mathematical theorems with pen-and-paper proofs, highlighting once again the rich interplay between automated and traditional mathematics.
\end{abstract}

\section{Introduction}\label{sec:intro}
Computation has played an increasingly large role within mathematics over the last 50 years. Back in 1976, Appel and Haken proved the celebrated \emph{Four Color Theorem} using a significant amount of computation~\cite{appelFourColorProblem1978}, which ultimately led to a formally verified \textsf{Coq} proof, written by Georges Gonthier in 2004~\cite{gonthierFourColourTheorem2008a, gonthier:hal-04034866}. 
 These results serve to highlight a dual role of computing in mathematics: solving problems and verifying solutions~\cite{avigad2023mathematics}. In present times, Large Language Models (LLMs) emerge as a new actor in computer-assisted mathematics; making progress in the \emph{cap-set problem}
 ~\cite{romera-paredesMathematicalDiscoveriesProgram2024}, solving olympiad-level geometry problems~\cite{trinhSolvingOlympiadGeometry2024a} and assisting formal theorem proving~\cite{yang2023leandojo}. Despite claims of progress in AI threatening mathematicians' jobs~\cite{eloundou2023gpts}, we adhere to the words of Jordan Ellenberg~\cite{castelvecchiDeepMindAIOutdoes2023}, co-author in the recent \emph{LLMs for the cap-set problem} article~\cite{romera-paredesMathematicalDiscoveriesProgram2024}: 
 \begin{center}
    \emph{``What's most exciting to me is modeling new modes of human–machine collaboration, [...] I don’t look to use these as a replacement for human mathematicians, but as a force multiplier.''} 
 \end{center}
In that spirit, this article presents a self-contained story of human-machine collaboration in mathematics, showcasing how automated reasoning tools can be incorporated in a mathematician's toolkit.
 
\paragraph{\bf Automated Reasoning.}
 Automated reasoning tools have been successfully used in the past to solve mathematical problems of diverse areas: Erd\H{o}s Discrepancy Conjecture~\cite{konev2014sat}, Keller's conjecture~\cite{brakensiek2023resolution}, the Packing Chromatic number of the infinite grid~\cite{Subercaseaux_Heule_2023}, and the Pythagorean Triples Problem~\cite{Heule_2016}, amongst many others.  Interestingly, before the recent progress made with LLMs~\cite{romera-paredesMathematicalDiscoveriesProgram2024}, the prior state of the art for the cap-set problem was obtained via SAT solving~\cite{tyrrellNewLowerBounds2023}. In the context of discrete geometry, Scheucher has used SAT solving to obtain state-of-the-art results in \emph{Erd\H{o}s-Szekeres} type problems~\cite{Scheucher_2023, Scheucher_2020, heuleHappyEndingEmpty2024}, making it our most closely related work.
The main novelty of this article is that we apply automated reasoning tools throughout the different stages of a mathematical problem: to guide the discovery of mathematical constructions, elicit a conjecture, and finally verify it until a certain bound to increase our confidence in it.

 \paragraph{\bf The Pentagon Minimization Problem.}

In 1933, Klein presented the following problem~\cite{Graham_Spencer_1990}: \emph{If five points lie on a plane, without three on a straight line, prove that four of the points will make a convex quadrilateral.} Klein's problem inspired two natural generalizations with a long lasting impact on combinatorial geometry:
\begin{problem}
For a given $k \geq 3$, is there always a minimum number of points  $n = g(k)$, such that any set of $n$ points in the plane, without three in a line, is guaranteed to contain $k$ points that are vertices of a convex $k$-gon?\label{problem:1}
\end{problem}

 \begin{problem}
 For a given $k \geq 3$, what is the minimum number of convex $k$-gons, $\mu_k(n)$, one can obtain after placing $n$ points in the plane without three in a line?\label{problem:2}
 \end{problem}

Erd\H{o}s and Szekeres published an affirmative answer to Problem~\ref{problem:1} in 1935.
Szekeres and Klein married shortly afterward, leading Erd\H{o}s to refer to Problem~\ref{problem:1} as the \emph{``Happy Ending Problem''}~\cite{Morris_Soltan_2000}. Problem~\ref{problem:2}, on the other hand, is directly mentioned for the first time by Erd\H{o}s and Guy in 1973~\cite{Erdos_Guy_1973}:
\emph{``More generally, one can ask for the least number of convex $k$-gons determined by $n$ points
 in the plane.}
A standard argument we show in~\Cref{sec:odd-even} implies that the limits
 \(
 	c_k := \lim_{n \to \infty} \mu_k(n) /{ \binom{n}{k}}
\)
are well-defined. 
Note that $c_3 = 1$ as every set of $3$ points in general position forms a triangle.  Perhaps surprisingly, $c_4$
is still unknown despite having received significant attention~\cite{Ongoing,AFMS}; the best known bounds are roughly
\(
	0.3799 \leq c_4 \leq 0.3804
\)~\cite{Ongoing,OngoingArxiv}.
Computation has played a crucial role in: (i) improving the bounds on $c_4$, (ii) computing $\mu_k(n)$ for small values of $k$ and $n$, and (iii) classifying small sets of points according to their geometric relationships~\cite{Aichholzer_Aurenhammer_Krasser_2001}.
As of today, the value of $\mu_4(n)$ is only known for $n \leq 27$ and $n=30$~\cite{AFMS}, and less is known for $k \geq 5$~\cite{Aichholzer_Aurenhammer_Krasser_2001,Aichholzer_Hackl_Vogtenhuber_2012,Aichholzer,goaoc_et_al:LIPIcs.SOCG.2015.300}, where the question raised by Erd\H{o}s and Guy, 50 years ago, remains widely open. 
We provide new insights into $\mu_5(n)$ and $c_5$, and opening directions for studying $\mu_k(n)$ and $c_k$ for larger values of $k$ as well.\\

\begin{table}[t]
	\caption{Improvements on $\mu_5(n)$. Brackets indicate the range of values in which $\mu_5(\cdot)$ was known to belong.}
	\centering
	\vspace{0.4em}
	\begin{tabular}{l|C{0.6cm}C{0.6cm}C{0.6cm}C{0.6cm}C{1.4cm}C{1.4cm}C{1.4cm}C{1.6cm}C{0.6cm}}
		\toprule
		\# of points $(n)$ & $\prescript{\leq}{}{8}$ & $9$ & $10$ & $11$ & $12$ & $13$ & $14$ & $15$ & $16$\\ \midrule
		Previously~\cite{Aichholzer} & 0 & 1 & 2 & 7 & [12, 13]
		 & [20, 34] & [40, 62] & [60, 113] & --- \\
		\textbf{Our work} & 0 & 1 & 2 & 7 & \textbf{12} & \textbf{27} & \textbf{42} & \textbf{77} & \textbf{112}\\
		\bottomrule
	\end{tabular}\label{table:improvements}
\end{table}

\vspace{-0.5em}
\paragraph{\bf Our contributions and outline.} 
This article makes progress on Problem~\ref{problem:2} in the particular case of $k=5$. 
As shown in~\Cref{table:improvements}, we fully determine $\mu_5(n)$ for $n \leq 16$, furthering prior results which reached $n = 11$~\cite{Aichholzer}. We start by providing some background into the geometry of the problem and known values of $\mu_5(n)$ in~\Cref{sec:background}, and continue in \Cref{sec:SLS} by presenting an initial exploration of $\mu_5(n)$ through \emph{Stochastic Local Search} (SLS) using \emph{signotopes}. 
Given the absence of open-source programs that find concrete point
realizations of signotopes, we present in~\Cref{sec:realizability} a simple local-search
program that finds them. We depict particular realizations obtained for the satisfying assignments that we found through SLS, which provide geometric insight into the problem. Based on those realizations, in~\Cref{sec:constructions} we propose and study two simple constructions providing a common upper bound of
\(
\mu_5(n) \leq \binom{\lfloor n/2\rfloor}{5} + \binom{\lceil n/2\rceil}{5},
\)  
which we conjecture to be optimal. \Cref{sec:odd-even} shows that if the conjecture holds for an odd value of $n$, then it will also hold for $n+1$.
\Cref{sec:maxsat} discusses the verification of the conjecture for $n \leq 16$ through MaxSAT, and finally,~\Cref{sec:conclusion} discusses the impact of the newly found values of $\mu_5(n)$ for bounding $c_5$ and offers a series of related open problems.	
Our code is available at \url{https://github.com/bsubercaseaux/minimize-5gons}.

\section{Background}\label{sec:background}
\input{sec-background.tex}

\section{Encoding and Stochastic Local Search}\label{sec:SLS}

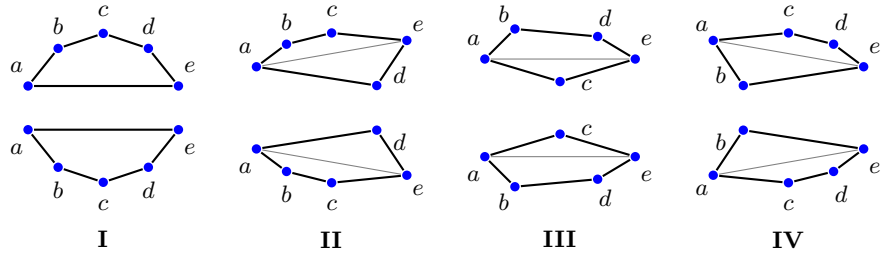
\begin{figure}[t]
\centering
\begin{minipage}{.24\textwidth}
\centering
\begin{tikzpicture}
\node at (-0.15,0.25) {$a$};
\node at (0.4,0.8) {$b$};
\node at (1,1.0) {$c$};
\node at (1.6,0.8) {$d$};
\node at (2.15,0.25) {$e$};
\node[inner sep=1pt, circle, thick, fill=blue, scale=1.25] (a) at (0,0) {};
\node[inner sep=1pt, circle, thick, fill=blue, scale=1.25] (b) at (0.4,0.5) {};
\node[inner sep=1pt, circle, thick, fill=blue, scale=1.25] (c) at (1,0.7) {};
\node[inner sep=1pt, circle, thick, fill=blue, scale=1.25] (d) at (1.6,0.5) {};
\node[inner sep=1pt, circle, thick, fill=blue, scale=1.25] (e) at (2,0) {};
\draw[black,thick] (a) -- (b) -- (c) -- (d) -- (e) -- (a);
\end{tikzpicture}

\bigskip

\begin{tikzpicture}
\node at (-0.15,-0.25) {$a$};
\node at (0.4,-0.8) {$b$};
\node at (1,-1.0) {$c$};
\node at (1.6,-0.8) {$d$};
\node at (2.15,-0.25) {$e$};
\node[inner sep=1pt, circle, thick, fill=blue, scale=1.25] (a) at (0,0) {};
\node[inner sep=1pt, circle, thick, fill=blue, scale=1.25] (b) at (0.4,-0.5) {};
\node[inner sep=1pt, circle, thick, fill=blue, scale=1.25] (c) at (1,-0.7) {};
\node[inner sep=1pt, circle, thick, fill=blue, scale=1.25] (d) at (1.6,-0.5) {};
\node[inner sep=1pt, circle, thick, fill=blue, scale=1.25] (e) at (2,0) {};
\draw[black,thick] (a) -- (b) -- (c) -- (d) -- (e) -- (a);
\end{tikzpicture}

{\bf I}
\end{minipage}
\begin{minipage}{.24\textwidth}
\centering
\begin{tikzpicture}
\node at (-0.15,0.25) {$a$};
\node at (0.4,0.6) {$b$};
\node at (1,0.75) {$c$};
\node at (1.9,-0.1) {$d$};
\node at (2.15,0.55) {$e$};
\node[inner sep=1pt, circle, thick, fill=blue, scale=1.25] (a) at (0,0) {};
\node[inner sep=1pt, circle, thick, fill=blue, scale=1.25] (b) at (0.4,0.3) {};
\node[inner sep=1pt, circle, thick, fill=blue, scale=1.25] (c) at (1,0.45) {};
\node[inner sep=1pt, circle, thick, fill=blue, scale=1.25] (d) at (1.6,-0.25) {};
\node[inner sep=1pt, circle, thick, fill=blue, scale=1.25] (e) at (2,0.35) {};
\draw[black,thick] (a) -- (b) -- (c) -- (e) -- (d) -- (a);
\draw[gray] (a) -- (e);
\end{tikzpicture}

\bigskip

\begin{tikzpicture}
\node at (-0.15,-0.25) {$a$};
\node at (0.4,-0.6) {$b$};
\node at (1,-0.75) {$c$};
\node at (1.9,0.1) {$d$};
\node at (2.15,-0.55) {$e$};
\node[inner sep=1pt, circle, thick, fill=blue, scale=1.25] (a) at (0,0) {};
\node[inner sep=1pt, circle, thick, fill=blue, scale=1.25] (b) at (0.4,-0.3) {};
\node[inner sep=1pt, circle, thick, fill=blue, scale=1.25] (c) at (1,-0.45) {};
\node[inner sep=1pt, circle, thick, fill=blue, scale=1.25] (d) at (1.6,0.25) {};
\node[inner sep=1pt, circle, thick, fill=blue, scale=1.25] (e) at (2,-0.35) {};
\draw[black,thick] (a) -- (b) -- (c) -- (e) -- (d) -- (a);
\draw[gray] (a) -- (e);
\end{tikzpicture}

{\bf II}
\end{minipage}
\begin{minipage}{.24\textwidth}
\centering
\begin{tikzpicture}
\node at (-0.15,0.25) {$a$};
\node at (0.25,0.6) {$b$};
\node at (1.35,-0.35) {$c$};
\node at (1.6,0.55) {$d$};
\node at (2.15,0.25) {$e$};
\node[inner sep=1pt, circle, thick, fill=blue, scale=1.25] (a) at (0,0) {};
\node[inner sep=1pt, circle, thick, fill=blue, scale=1.25] (b) at (0.4,0.4) {};
\node[inner sep=1pt, circle, thick, fill=blue, scale=1.25] (c) at (1,-0.3) {};
\node[inner sep=1pt, circle, thick, fill=blue, scale=1.25] (d) at (1.5,0.3) {};
\node[inner sep=1pt, circle, thick, fill=blue, scale=1.25] (e) at (2,0) {};
\draw[black,thick] (a) -- (b) -- (d) -- (e) -- (c) -- (a);
\draw[gray] (a) -- (e);
\end{tikzpicture}

\medskip

\begin{tikzpicture}
\node at (-0.15,-0.25) {$a$};
\node at (0.25,-0.6) {$b$};
\node at (1.35,0.35) {$c$};
\node at (1.6,-0.55) {$d$};
\node at (2.15,-0.25) {$e$};
\node[inner sep=1pt, circle, thick, fill=blue, scale=1.25] (a) at (0,-0) {};
\node[inner sep=1pt, circle, thick, fill=blue, scale=1.25] (b) at (0.4,-0.4) {};
\node[inner sep=1pt, circle, thick, fill=blue, scale=1.25] (c) at (1,0.3) {};
\node[inner sep=1pt, circle, thick, fill=blue, scale=1.25] (d) at (1.5,-0.3) {};
\node[inner sep=1pt, circle, thick, fill=blue, scale=1.25] (e) at (2,0) {};
\draw[black,thick] (a) -- (b) -- (d) -- (e) -- (c) -- (a);
\draw[gray] (a) -- (e);
\end{tikzpicture}

{\bf III}
\end{minipage}
\begin{minipage}{.24\textwidth}
\centering
\begin{tikzpicture}
\node at (-0.15,0.55) {$a$};
\node at (0.1,-0.1) {$b$};
\node at (1,0.75) {$c$};
\node at (1.7,0.6) {$d$};
\node at (2.15,0.25) {$e$};
\node[inner sep=1pt, circle, thick, fill=blue, scale=1.25] (a) at (0,0.35) {};
\node[inner sep=1pt, circle, thick, fill=blue, scale=1.25] (b) at (0.4,-0.25) {};
\node[inner sep=1pt, circle, thick, fill=blue, scale=1.25] (c) at (1,0.45) {};
\node[inner sep=1pt, circle, thick, fill=blue, scale=1.25] (d) at (1.6,0.3) {};
\node[inner sep=1pt, circle, thick, fill=blue, scale=1.25] (e) at (2,0) {};
\draw[black,thick] (a) -- (b) -- (e) -- (d) -- (c) -- (a);
\draw[gray] (a) -- (e);
\end{tikzpicture}

\bigskip

\begin{tikzpicture}
\node at (-0.15,-0.55) {$a$};
\node at (0.1,0.1) {$b$};
\node at (1,-0.75) {$c$};
\node at (1.7,-0.6) {$d$};
\node at (2.15,-0.25) {$e$};
\node[inner sep=1pt, circle, thick, fill=blue, scale=1.25] (a) at (0,-0.35) {};
\node[inner sep=1pt, circle, thick, fill=blue, scale=1.25] (b) at (0.4,0.25) {};
\node[inner sep=1pt, circle, thick, fill=blue, scale=1.25] (c) at (1,-0.45) {};
\node[inner sep=1pt, circle, thick, fill=blue, scale=1.25] (d) at (1.6,-0.3) {};
\node[inner sep=1pt, circle, thick, fill=blue, scale=1.25] (e) at (2,0) {};
\draw[black,thick] (a) -- (b) -- (e) -- (d) -- (c) -- (a);
\draw[gray] (a) -- (e);
\end{tikzpicture}

{\bf IV}
\end{minipage}
\caption{The convex 5-gon cases based on the position of $b$, $c$, $d$ w.r.t.\ the line $\overline{ae}$.}\label{fig:cases}
\end{figure}


We will use the signotopes $\sigma(a, b, c)$ directly as propositional variables in our encoding. As a first step, we directly add the $O(n^4)$ axiom clauses of Equations~(\ref{eq:axiom1})-(\ref{eq:axiom4}).
Then, in order to minimize the number of induced convex $5$-gons, we use an idea of Szekeres and Peters~\cite{szekeres_peters_2006}.  Szekeres and Peters identified the four cases that form a convex $5$-gon, depending on where the three middle points $b$, $c$, and $d$ are located with respect to the line through the leftmost point $a$ and the rightmost point $e$:
%
\begin{itemize}
\item[~~]  \textbf{Case I}:\; \quad $\orientp{a}{b}{c} = \phantom{\lnot}\orientp{b}{c}{d} = \phantom{\lnot}\orientp{c}{d}{e}$
\item[~~]  \textbf{Case II}:\;  \;\,$\orientp{a}{b}{c} = \phantom{\lnot}\orientp{b}{c}{e} = \orientn{a}{d}{e}$
\item[~~]  \textbf{Case III}:\; $\orientp{a}{b}{d} = \phantom{\lnot}\orientp{b}{d}{e} = \orientn{a}{c}{e}$
\item[~~]  \textbf{Case IV}:\; $\orientp{a}{b}{e} = \orientn{a}{c}{d} = \orientn{c}{d}{e}$.
\end{itemize}

The four cases are illustrated in Figure~\ref{fig:cases}, showcasing that each case has two possible orientations depending on the value of the signotopes that the case asserts to be equal.
Given that we will first use \emph{Stochastic Local Search} (SLS), it is important to recall that an SLS solver attempts to find an assignment that minimizes the number of falsified clauses in its input formula, without any guarantee of optimality except for when a fully satisfying assignment is found. 
Therefore, as we want to find assignments to the signotope variables that minimize the number of convex 5-gons, we desire an encoding where each convex 5-gon falsifies exactly one clause. 
We achieve this by adding the following clauses, based on the disjoint cases \textbf{I}-\textbf{IV}, for every sorted tuple of five points $(a,b,c,d,e)$:


\vspace{-10pt}
\begin{alignat}{5}
	\orientp{a}{b}{c} &\; \lor& \orientp{b}{c}{d} &\;\lor& \orientp{c}{d}{e} \label{eq:hard1}\\ 
	\orientn{a}{b}{c} &\; \lor& \orientn{b}{c}{d} &\;\lor& \orientn{c}{d}{e} \label{eq:hard2}\\ 
	\orientp{a}{b}{c} &\; \lor& \orientp{b}{c}{e} &\;\lor& \orientn{a}{d}{e} \label{eq:hard3}\\ 
	\orientn{a}{b}{c} &\; \lor& \orientn{b}{c}{e} &\;\lor& \orientp{a}{d}{e} \label{eq:hard4}\\ 
	\orientp{a}{b}{d} &\; \lor& \orientp{b}{d}{e} &\;\lor& \orientn{a}{c}{e} \label{eq:hard5}\\ 
	\orientn{a}{b}{d} &\; \lor& \orientn{b}{d}{e} &\;\lor& \orientp{a}{c}{e} \label{eq:hard6}\\ 
	\orientp{a}{b}{e} &\; \lor& \orientn{a}{c}{d} &\;\lor& \orientn{c}{d}{e} \label{eq:hard7}\\ 
	\orientn{a}{b}{e} &\; \lor& \orientp{a}{c}{d} &\;\lor& \orientp{c}{d}{e} \label{eq:hard8}.
\end{alignat}

Note that the best assignments found through SLS might, in principle, violate the signotope axiom clauses of Equations~(\ref{eq:axiom1})-(\ref{eq:axiom4}) in order to minimize the number of falsified clauses. Interestingly, we tested for about a thousand best assignments whether any axiom clauses were
falsified and this was never the case. Therefore the number of falsified clauses in all the best assignments found through SLS, presented in~\Cref{table:sls}, would constitute an upper bound on the minimum number of convex 5-gons if they were \emph{realizable}.

We also experimented with formulas without the signotope axiom clauses. The best number of falsified clauses of these formulas match numbers on ~\Cref{table:sls}.
So potentially violating many axiom clauses does not result in fewer 5-gons. It is therefore not surprising that none of the axiom clauses were falsified in the best found assignments. However, the runtimes to obtain the best-known values were substantially higher (roughly an order of magnitude) for the formulas without the axiom clauses. So these clauses are helpful to reduce the runtime.



In terms of software, we tested all the algorithms in {\sf UBCSAT}~\cite{ubcsat} and the DDFW algorithm~\cite{ddfw} turned out to have the best performance.
We ran {\sf UBCSAT} with its default settings, which means it restarts every 100\,000 flips. For some of the harder formulas
this resulted in hundreds of restarts. An interesting observation is that the optimal assignments are harder to find when $n$ is even. 
Observe that best number of falsified clauses matches exactly the conjectured values, apart from $n=30$ where the best found assignment after 12 hours is 
1 above the conjectured value, suggesting that we reached the limit of SLS for this problem.



\begin{table}[t]
	\caption{SLS results of formulas for $\mu_5(n)$ showing the best number of falsified clauses and the time, in seconds, to find that bound.}
	\centering
	\begin{tabular}{@{~~}c@{~~}|@{~~}c@{~~}c@{~~}c@{~~}c@{~~}c@{~~}c@{~~}c@{~~}c@{~~}c@{~~}c@{~~}c@{~~}c@{~~}c@{~~}}
		\toprule
		$n$ & $9$ & $10$ & $11$ & $12$ & $13$ & $14$ & $15$ & $16$ & $17$ & $18$ & $19$ & $20$\\ \midrule
		best & 1 & 2 & 7 & 12 & 27 & 42 & 77 & 112 & 182 & 252 & 378 & 504 \\
		time [s] & 0.00 & 0.00 & 0.00 & 0.00 & 0.01 & 0.01 & 0.01 & 0.02 & 0.02 & 2.03 & 0.94 & 174.11\\
		\bottomrule
	\end{tabular}
	
\vspace{10pt}
	
	\begin{tabular}{@{~~}c@{~~}|@{~~}c@{~~}c@{~~}c@{~~}c@{~~}c@{~~}c@{~~}c@{~~}c@{~~}c@{~~}c@{~~}c@{~~}}
		\toprule
		$n$ & $21$ & $22$ & $23$ & $24$ & $25$ & $26$ & $27$ & $28$ & $29$ & $30$\\ \midrule
		best & 714 & 924 & 1254 & 1584 & 2079 & 2574 & 3289 & 4004 & 5005 & 6007 \\
		time [s] & 3.34& 43.92 & 11.64 & 472.33 & 63.48 & 5268.1 & 1555.5 & 1791.9& 467.36 & 18\,244\\
		\bottomrule
	\end{tabular}\label{table:sls}
\end{table}

\section{Realizability}\label{sec:realizability}

Deciding whether a signotope assignment can be realized by a set of points in the plane is a hard combinatorial problem, complete for the complexity class $\exists\mathbb{R}$, which satisfies $\mathrm{NP} \subseteq \exists\mathbb{R} \subseteq\mathrm{PSPACE}$~\cite{Shor}.
To the best of our knowledge, no open-source tools for the realizability problem are publicly available. Because of this, we present as a contribution of independent interest, a simple local search approach to realizability, that has proved effective in this problem.
 We use \texttt{LocalSolver v12.0}~\cite{Gardi2014}, a local search engine that supports floating point variables and provides free academic licenses, in order to find realizations for small numbers of points $(n \leq 16)$.
 Concretely, the problem at hand consists of, given a (not necessarily consistent) assignment to the signotope variables $\sigma(a, b, c), \ldots$, find a set of points $(x_a, y_a), (x_b, y_b), (x_c, y_c), \ldots$ that satisfies them all. We add one constraint per signotope variable, and maximize the minimum distance between any pair of points in order to \emph{regularize} the obtained realizations and avoid floating-point-arithmetic issues. 
 Experimentally, we found that if no constraint is placed on the distance between points, \texttt{LocalSolver} tends towards solutions that place all points extremely close to each other, or in a straight line, where the constraints will be satisfied only due to floating point imprecision.
Therefore, our model consists of\footnote{Constraints are over all indices $a, b, c \in \{1, \ldots, n\}$, with $a < b < c$ in case more than one index is involved.}:
\begin{framed}
\begin{align*}
	\textbf{maximize}  \quad z &\\
	\textbf{subject to}  \quad z & \leq \sqrt{(x_{a} - x_b)^2 + (y_a - y_b)^2}, \\ 
	\quad \varepsilon  &  <\sigma^\star(a,b,c) \cdot [(y_c - y_a)(x_b - x_a) - (x_c - x_a)(y_b - y_a)],\\
	0 &\leq x_{a} \leq K, \\ 
	 0 &\leq y_a \leq K, \\ 
\end{align*}
\end{framed}

\noindent where $\sigma^\star(a, b, c) = 1$ if $\sigma(a, b, c)$  and $-1$ otherwise. The parameters $\varepsilon$ and $K$ are experimentally determined, and also contribute towards avoiding \emph{degenerate} solutions due to floating point arithmetic. In particular, setting a value of $\varepsilon$ that is too close to $0$ (say, $\varepsilon = 10^{-10}$) allows for degenerate solutions that only satisfy constraints due to floating-point-arithmetic quirks, whereas setting e.g., $\varepsilon = 10^{-3}$ may result in an unfeasible set of constraints even if the signotope assignment is realizable.
By implementing this model in~\texttt{LocalSolver}, we obtained the realizations depicted in~\Cref{fig:realization-pinwheel} and~\Cref{fig:realization-parabolic}. 
We remark that, consistently with the asymptotic result stating that most signotope assignments are \emph{not} realizable~\cite{Bjorner_Las_Vergnas_Sturmfels_White_Ziegler_1999,Shor}, only about 5\% of the 10\,000 different assignments we found through SLS for $n \in \{12, 14, 16\}$ led to realizations (in under 100s). 
\begin{figure}
	\centering
	\begin{subfigure}{0.45\textwidth}
		\centering
		\begin{tikzpicture}[scale=0.8]
			\begin{axis}[
				xmin=0,xmax=10,
				ymin=0,ymax=10,
				grid=both,
				grid style={line width=.1pt, draw=gray!10},
				major grid style={line width=.2pt,draw=gray!50},
				axis lines=middle,
				minor tick num=5,
				enlargelimits={abs=0.5},
				axis line style={latex-latex},
				ticklabel style={font=\tiny,fill=white},
				xlabel style={at={(ticklabel* cs:1)},anchor=north west},
				ylabel style={at={(ticklabel* cs:1)},anchor=south west}
			]
			\node[inner sep=1pt, circle, thick, fill=blue, scale=1.25] (0) at (axis cs:0.030294298034345045, 0.04931397512102001) {};
			\node[inner sep=1pt, circle, thick, fill=blue, scale=1.25] (1) at (axis cs:0.07899545114343348, 9.945540802574182) {};
			\node[inner sep=1pt, circle, thick, fill=blue, scale=1.25] (2) at (axis cs:0.7591671583605977, 1.0044792058518814) {};
			\node[inner sep=1pt, circle, thick, fill=blue, scale=1.25] (3) at (axis cs:0.9212051954636317, 9.111754169053281) {};
			\node[inner sep=1pt, circle, thick, fill=blue, scale=1.25] (4) at (axis cs:1.8491042254725392, 7.706062957498922) {};
			\node[inner sep=1pt, circle, thick, fill=blue, scale=1.25] (5) at (axis cs:5.274267803605714, 8.111199009860846) {};
			\node[inner sep=1pt, circle, thick, fill=blue, scale=1.25] (6) at (axis cs:5.30632067190972, 2.0436696186794925) {};
			\node[inner sep=1pt, circle, thick, fill=blue, scale=1.25] (7) at (axis cs:8.344774347606368, 2.49760898505815) {};
			\node[inner sep=1pt, circle, thick, fill=blue, scale=1.25] (8) at (axis cs:9.161405320253825, 8.87947420201652) {};
			\node[inner sep=1pt, circle, thick, fill=blue, scale=1.25] (9) at (axis cs:9.260791812948357, 1.0148059628612955) {};
			\node[inner sep=1pt, circle, thick, fill=blue, scale=1.25] (10) at (axis cs:9.937351727262362, 0.04266311256750367) {};
			\node[inner sep=1pt, circle, thick, fill=blue, scale=1.25] (11) at (axis cs:9.967084633256029, 9.95924727471215) {};
			\draw[dashed] (0) -- (2) -- (6);
			\draw[dashed] (7) -- (9) -- (10);
			\draw[dashed] (1) -- (3) -- (4);
			\draw[dashed] -- (5) -- (8) -- (11);
			\end{axis}
		\end{tikzpicture}
	\caption{A realization obtained for the problem instance $\mu_5(12) = 12$ that inspires the~\emph{pinwheel construction}.}
	\label{fig:realization-pinwheel}
	\end{subfigure}
	\hfill
	\begin{subfigure}{0.45\textwidth}
		\centering
		\begin{tikzpicture}[scale=0.8]
			\begin{axis}[
				xmin=0,xmax=6,
				ymin=0,ymax=10,
				grid=both,
				grid style={line width=.1pt, draw=gray!10},
				major grid style={line width=.2pt,draw=gray!50},
				axis lines=middle,
				minor tick num=5,
				enlargelimits={abs=0.5},
				axis line style={latex-latex},
				ticklabel style={font=\tiny,fill=white},
				xlabel style={at={(ticklabel* cs:1)},anchor=north west},
				ylabel style={at={(ticklabel* cs:1)},anchor=south west}
			]
			
			\node[inner sep=1pt, circle, thick, fill=blue, scale=1.25] (0) at (axis cs:0.09323182225641637, 6.477855528985496) {};
			\node[inner sep=1pt, circle, thick, fill=blue, scale=1.25] (1) at (axis cs:1.620088260004434, 9.999998772471152) {};
			\node[inner sep=1pt, circle, thick, fill=blue, scale=1.25] (2) at (axis cs:1.6596751176894275, 5.922476584812963) {};
			\node[inner sep=1pt, circle, thick, fill=blue, scale=1.25] (3) at (axis cs:2.1406444694008675, 8.54746826965688) {};
			\node[inner sep=1pt, circle, thick, fill=blue, scale=1.25] (4) at (axis cs:2.6608857055560406, 7.09571344851647) {};
			\node[inner sep=1pt, circle, thick, fill=blue, scale=1.25] (5) at (axis cs:3.184473827891598, 5.644365240201632) {};
			\node[inner sep=1pt, circle, thick, fill=blue, scale=1.25] (6) at (axis cs:3.184483311988519, 4.0955688923297355) {};
			\node[inner sep=1pt, circle, thick, fill=blue, scale=1.25] (7) at (axis cs:4.163218713619542, 2.9039293209943935) {};
			\node[inner sep=1pt, circle, thick, fill=blue, scale=1.25] (8) at (axis cs:4.5397328695711785, 4.901343442269086) {};
			\node[inner sep=1pt, circle, thick, fill=blue, scale=1.25] (9) at (axis cs:4.683411028572396, 1.4522791761844014) {};
			\node[inner sep=1pt, circle, thick, fill=blue, scale=1.25] (10) at (axis cs:5.203829185711722, 2.026775433462919e-05) {};
			\node[inner sep=1pt, circle, thick, fill=blue, scale=1.25] (11) at (axis cs:6.0318920497203585, 4.372302451332322) {};

			\draw[dashed] (0) -- (2) -- (6) -- (7) -- (9) -- (10);
			\draw[dashed] (1) -- (3) -- (4) -- (5) -- (8) -- (11);
			\end{axis}
		\end{tikzpicture}
	\caption{A different realization obtained for the problem instance $\mu_5(12) = 12$ that inspires the~\emph{parabolic construction}.}
	\label{fig:realization-parabolic}
	\end{subfigure}
	
\caption{Illustration of two different realizations obtained with \texttt{LocalSolver}. Both realizations come from a single signotope orientation obtained through SLS under different executions.
 Dashed lines are for illustrative purposes only, in order to showcase the similarity with~\Cref{fig:construction-1} and~\Cref{fig:construction-2}.}
\end{figure}
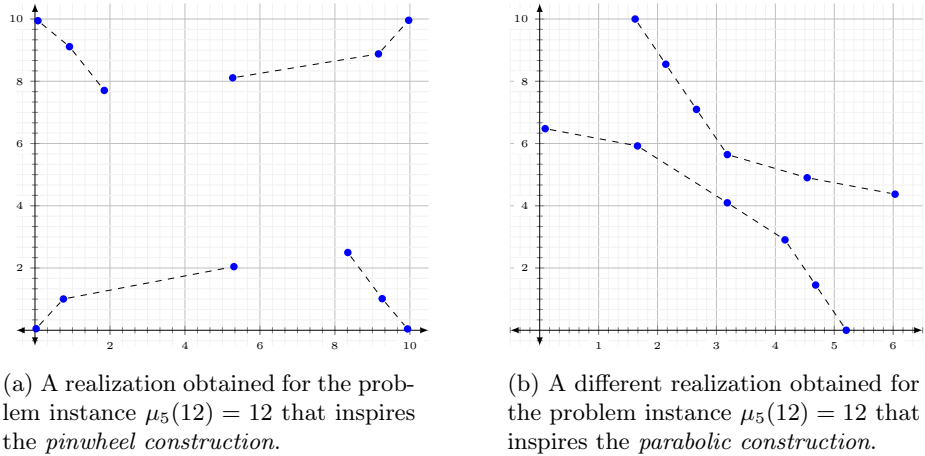

\section{Constructions}\label{sec:constructions}
We present two different constructions achieving a common bound: (i) the \emph{pinwheel construction} in~\Cref{subsec:pinwheel}, which generalizes the realization depicted in~\Cref{fig:realization-pinwheel}, and (ii) the \emph{parabolic construction} in~\Cref{subsec:parabolic}, which generalizes the realization depicted in~\Cref{fig:realization-parabolic}.

\subsection{The Pinwheel Construction}\label{subsec:pinwheel}

This construction, illustrated in~\Cref{fig:construction-1} requires the number of points $n=4k$ to be a multiple of $4$. It consists of four \emph{spokes} defined as follows: 
\begin{align*}
	S_1 &= \{ (k+j, -j(k-j)/k^3 + 1) \mid j \in \{0, 1, \ldots, k-1\}\} \\
	S_2 &= \{ ( -(k+j), j(k-j)/k^3 - 1) \mid j \in \{0, 1, \ldots, k-1\}\} \\
	S_3 &= \{ ( j(k-j)/k^3 - 1, k+j) \mid j \in \{0, 1, \ldots, k-1\}\} \\
	S_4 &= \{ ( -j(k-j)/k^3 + 1, -(k+j)) \mid j \in \{0, 1, \ldots, k-1\}\}.
\end{align*}

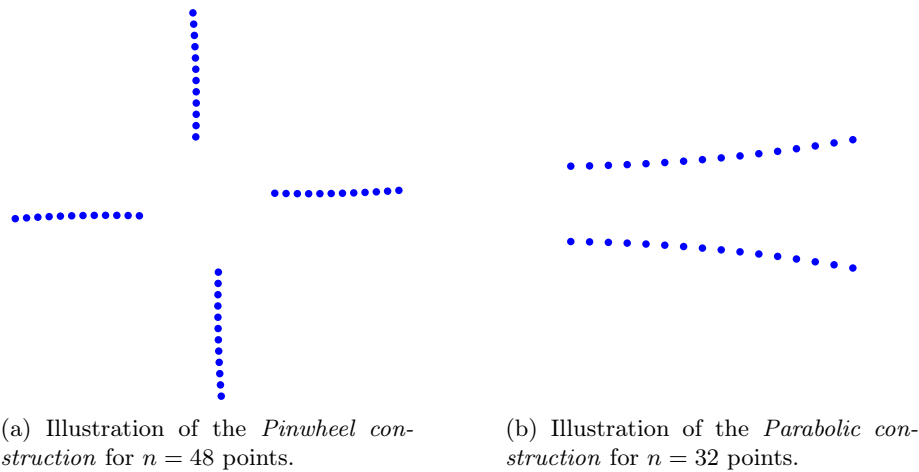
\begin{figure}[ht]
	\begin{subfigure}{0.45\textwidth}
	\centering
			\begin{tikzpicture}[scale=0.15]
				\node[inner sep=1pt, circle, thick, fill=blue, scale=1] at (6.0, 1.0) {};
				\node[inner sep=1pt, circle, thick, fill=blue, scale=1] at (-6.0, -1.0) {};
				\node[inner sep=1pt, circle, thick, fill=blue, scale=1] at (-1.0, 6.0) {};
				\node[inner sep=1pt, circle, thick, fill=blue, scale=1] at (1.0, -6.0) {};
				\node[inner sep=1pt, circle, thick, fill=blue, scale=1] at (7.0, 0.9768518518518519) {};
				\node[inner sep=1pt, circle, thick, fill=blue, scale=1] at (-7.0, -0.9768518518518519) {};
				\node[inner sep=1pt, circle, thick, fill=blue, scale=1] at (-0.9768518518518519, 7.0) {};
				\node[inner sep=1pt, circle, thick, fill=blue, scale=1] at (0.9768518518518519, -7.0) {};
				\node[inner sep=1pt, circle, thick, fill=blue, scale=1] at (8.0, 0.962962962962963) {};
				\node[inner sep=1pt, circle, thick, fill=blue, scale=1] at (-8.0, -0.962962962962963) {};
				\node[inner sep=1pt, circle, thick, fill=blue, scale=1] at (-0.962962962962963, 8.0) {};
				\node[inner sep=1pt, circle, thick, fill=blue, scale=1] at (0.962962962962963, -8.0) {};
				\node[inner sep=1pt, circle, thick, fill=blue, scale=1] at (9.0, 0.9583333333333334) {};
				\node[inner sep=1pt, circle, thick, fill=blue, scale=1] at (-9.0, -0.9583333333333334) {};
				\node[inner sep=1pt, circle, thick, fill=blue, scale=1] at (-0.9583333333333334, 9.0) {};
				\node[inner sep=1pt, circle, thick, fill=blue, scale=1] at (0.9583333333333334, -9.0) {};
				\node[inner sep=1pt, circle, thick, fill=blue, scale=1] at (10.0, 0.962962962962963) {};
				\node[inner sep=1pt, circle, thick, fill=blue, scale=1] at (-10.0, -0.962962962962963) {};
				\node[inner sep=1pt, circle, thick, fill=blue, scale=1] at (-0.962962962962963, 10.0) {};
				\node[inner sep=1pt, circle, thick, fill=blue, scale=1] at (0.962962962962963, -10.0) {};
				\node[inner sep=1pt, circle, thick, fill=blue, scale=1] at (11.0, 0.9768518518518519) {};
				\node[inner sep=1pt, circle, thick, fill=blue, scale=1] at (-11.0, -0.9768518518518519) {};
				\node[inner sep=1pt, circle, thick, fill=blue, scale=1] at (-0.9768518518518519, 11.0) {};
				\node[inner sep=1pt, circle, thick, fill=blue, scale=1] at (0.9768518518518519, -11.0) {};
				\node[inner sep=1pt, circle, thick, fill=blue, scale=1] at (12.0, 1.0) {};
				\node[inner sep=1pt, circle, thick, fill=blue, scale=1] at (-12.0, -1.0) {};
				\node[inner sep=1pt, circle, thick, fill=blue, scale=1] at (-1.0, 12.0) {};
				\node[inner sep=1pt, circle, thick, fill=blue, scale=1] at (1.0, -12.0) {};
				\node[inner sep=1pt, circle, thick, fill=blue, scale=1] at (13.0, 1.0324074074074074) {};
				\node[inner sep=1pt, circle, thick, fill=blue, scale=1] at (-13.0, -1.0324074074074074) {};
				\node[inner sep=1pt, circle, thick, fill=blue, scale=1] at (-1.0324074074074074, 13.0) {};
				\node[inner sep=1pt, circle, thick, fill=blue, scale=1] at (1.0324074074074074, -13.0) {};
				\node[inner sep=1pt, circle, thick, fill=blue, scale=1] at (14.0, 1.074074074074074) {};
				\node[inner sep=1pt, circle, thick, fill=blue, scale=1] at (-14.0, -1.074074074074074) {};
				\node[inner sep=1pt, circle, thick, fill=blue, scale=1] at (-1.074074074074074, 14.0) {};
				\node[inner sep=1pt, circle, thick, fill=blue, scale=1] at (1.074074074074074, -14.0) {};
				\node[inner sep=1pt, circle, thick, fill=blue, scale=1] at (15.0, 1.125) {};
				\node[inner sep=1pt, circle, thick, fill=blue, scale=1] at (-15.0, -1.125) {};
				\node[inner sep=1pt, circle, thick, fill=blue, scale=1] at (-1.125, 15.0) {};
				\node[inner sep=1pt, circle, thick, fill=blue, scale=1] at (1.125, -15.0) {};
				\node[inner sep=1pt, circle, thick, fill=blue, scale=1] at (16.0, 1.1851851851851851) {};
				\node[inner sep=1pt, circle, thick, fill=blue, scale=1] at (-16.0, -1.1851851851851851) {};
				\node[inner sep=1pt, circle, thick, fill=blue, scale=1] at (-1.1851851851851851, 16.0) {};
				\node[inner sep=1pt, circle, thick, fill=blue, scale=1] at (1.1851851851851851, -16.0) {};
				\node[inner sep=1pt, circle, thick, fill=blue, scale=1] at (17.0, 1.2546296296296298) {};
				\node[inner sep=1pt, circle, thick, fill=blue, scale=1] at (-17.0, -1.2546296296296298) {};
				\node[inner sep=1pt, circle, thick, fill=blue, scale=1] at (-1.2546296296296298, 17.0) {};
				\node[inner sep=1pt, circle, thick, fill=blue, scale=1] at (1.2546296296296298, -17.0) {};
			\end{tikzpicture}
			\caption{ Illustration of the \emph{Pinwheel construction} for $n=48$ points.}
			\label{fig:construction-1}
		\end{subfigure}
		\hfill
		\begin{subfigure}{0.45\textwidth}
			\centering
			\begin{tikzpicture}[scale=0.25]
		\node[inner sep=1pt, circle, thick, fill=blue, scale=1] at (1.0, 2.00552427172802) {};
		\node[inner sep=1pt, circle, thick, fill=blue, scale=1] at (1.0, -2.00552427172802) {};
		\node[inner sep=1pt, circle, thick, fill=blue, scale=1] at (2.0, 2.02209708691208) {};
		\node[inner sep=1pt, circle, thick, fill=blue, scale=1] at (2.0, -2.02209708691208) {};
		\node[inner sep=1pt, circle, thick, fill=blue, scale=1] at (3.0, 2.049718445552179) {};
		\node[inner sep=1pt, circle, thick, fill=blue, scale=1] at (3.0, -2.049718445552179) {};
		\node[inner sep=1pt, circle, thick, fill=blue, scale=1] at (4.0, 2.0883883476483183) {};
		\node[inner sep=1pt, circle, thick, fill=blue, scale=1] at (4.0, -2.0883883476483183) {};
		\node[inner sep=1pt, circle, thick, fill=blue, scale=1] at (5.0, 2.1381067932004973) {};
		\node[inner sep=1pt, circle, thick, fill=blue, scale=1] at (5.0, -2.1381067932004973) {};
		\node[inner sep=1pt, circle, thick, fill=blue, scale=1] at (6.0, 2.1988737822087163) {};
		\node[inner sep=1pt, circle, thick, fill=blue, scale=1] at (6.0, -2.1988737822087163) {};
		\node[inner sep=1pt, circle, thick, fill=blue, scale=1] at (7.0, 2.270689314672975) {};
		\node[inner sep=1pt, circle, thick, fill=blue, scale=1] at (7.0, -2.270689314672975) {};
		\node[inner sep=1pt, circle, thick, fill=blue, scale=1] at (8.0, 2.353553390593274) {};
		\node[inner sep=1pt, circle, thick, fill=blue, scale=1] at (8.0, -2.353553390593274) {};
		\node[inner sep=1pt, circle, thick, fill=blue, scale=1] at (9.0, 2.447466009969612) {};
		\node[inner sep=1pt, circle, thick, fill=blue, scale=1] at (9.0, -2.447466009969612) {};
		\node[inner sep=1pt, circle, thick, fill=blue, scale=1] at (10.0, 2.5524271728019903) {};
		\node[inner sep=1pt, circle, thick, fill=blue, scale=1] at (10.0, -2.5524271728019903) {};
		\node[inner sep=1pt, circle, thick, fill=blue, scale=1] at (11.0, 2.6684368790904083) {};
		\node[inner sep=1pt, circle, thick, fill=blue, scale=1] at (11.0, -2.6684368790904083) {};
		\node[inner sep=1pt, circle, thick, fill=blue, scale=1] at (12.0, 2.795495128834866) {};
		\node[inner sep=1pt, circle, thick, fill=blue, scale=1] at (12.0, -2.795495128834866) {};
		\node[inner sep=1pt, circle, thick, fill=blue, scale=1] at (13.0, 2.9336019220353635) {};
		\node[inner sep=1pt, circle, thick, fill=blue, scale=1] at (13.0, -2.9336019220353635) {};
		\node[inner sep=1pt, circle, thick, fill=blue, scale=1] at (14.0, 3.0827572586919008) {};
		\node[inner sep=1pt, circle, thick, fill=blue, scale=1] at (14.0, -3.0827572586919008) {};
		\node[inner sep=1pt, circle, thick, fill=blue, scale=1] at (15.0, 3.242961138804478) {};
		\node[inner sep=1pt, circle, thick, fill=blue, scale=1] at (15.0, -3.242961138804478) {};
		\node[inner sep=1pt, circle, thick, fill=blue, scale=1] at (16.0, 3.414213562373095) {};
		\node[inner sep=1pt, circle, thick, fill=blue, scale=1] at (16.0, -3.414213562373095) {};
		\node[] at (5, -10) {};		
	\end{tikzpicture}
			\caption{Illustration of the \emph{Parabolic construction} for $n = 32$ points.}\label{fig:construction-2}
		\end{subfigure}
	\caption{Illustration of the constructions achieving the bound of~\Cref{thm:upper-bounds}.}\label{fig:constructions}
\end{figure} 


\begin{proposition}
	The number of convex $5$-gons obtained by applying the pinwheel construction on $n = 4k$ points is exactly $2\binom{2k}{5}$.
\end{proposition}

\begin{proof}[Proof sketch, illustrated in~\Cref{fig:construction-1-proof-sketch}]
    The proof is by cases according to which spokes contain the $5$ points of an arbitrary pentagon.
     If all five points are contained in the same spoke $S_i$, then the pentagon is convex due to the curvature of the spokes.
     If four points are contained in a spoke $S_i$, then to make a convex $5$-gon the fifth point must be in the next spoke counterclockwise, $S_{(i+1)\%4}$. Similarly, if a spoke contains 3 points, and the next spoke counterclockwise contains the remaining $2$ points, then the pentagon is convex. Finally, every other case yields a non-convex $5$-gon. 
     As there are $4$ ways of choosing the unique spoke with the largest number of points in a convex $5$-gon, this results in 
     \[
        4 \cdot \left( \binom{k}{5} + \binom{k}{4} \cdot \binom{k}{1} + \binom{k}{3} \cdot \binom{k}{2} \right) = 2 \binom{2k}{5},	
    \]
    where the equality follows by Vandermonde's identity.
    \end{proof}

\input{sketch-proof-pinwheel}

\begin{figure}
\begin{subfigure}{0.3\textwidth}
	\begin{tikzpicture}[scale=0.3]
\node[inner sep=1pt, circle, thick, fill=blue, scale=1] at (1, 2.05) {};
\node[inner sep=1pt, circle, thick, fill=blue, scale=1] at (1, -2.05) {};
\node[inner sep=1pt, circle, thick, fill=blue, scale=1] at (2, 2.2) {};
\node[inner sep=1pt, circle, thick, fill=blue, scale=1] at (2, -2.2) {};
\node[inner sep=1pt, circle, thick, fill=blue, scale=1] at (3, 2.45) {};
\node[inner sep=1pt, circle, thick, fill=blue, scale=1] at (3, -2.45) {};
\node[inner sep=1pt, circle, thick, fill=blue, scale=1] at (4, 2.8) {};
\node[inner sep=1pt, circle, thick, fill=blue, scale=1] at (4, -2.8) {};
\node[inner sep=1pt, circle, thick, fill=blue, scale=1] at (5, 3.25) {};
\node[inner sep=1pt, circle, thick, fill=blue, scale=1] at (5, -3.25) {};
\node[inner sep=1pt, circle, thick, fill=blue, scale=1] at (6, 3.8) {};
\node[inner sep=1pt, circle, thick, fill=blue, scale=1] at (6, -3.8) {};
\node[inner sep=1pt, circle, thick, fill=blue, scale=1] at (7, 4.45) {};
\node[inner sep=1pt, circle, thick, fill=blue, scale=1] at (7, -4.45) {};
\node[inner sep=1pt, circle, thick, fill=blue, scale=1] at (8, 5.2) {};
\node[inner sep=1pt, circle, thick, fill=blue, scale=1] at (8, -5.2) {};
\node[inner sep=1pt, circle, thick, fill=blue, scale=1] at (9, 6.05) {};
\node[inner sep=1pt, circle, thick, fill=blue, scale=1] at (9, -6.05) {};
\node[inner sep=1pt, circle, thick, fill=blue, scale=1] at (10, 7.0) {};
\node[inner sep=1pt, circle, thick, fill=blue, scale=1] at (10, -7.0) {};
\node[draw, circle, fill=green, inner sep=0pt, minimum width=2mm] at (1, -2.05) {};
\coordinate (a) at (1, -2.05);
\node[above, yshift=2mm] at (a) {$a$};
\node[draw, circle, fill=green, inner sep=0pt, minimum width=2mm] at (3, -2.45) {};
\coordinate (b) at (3, -2.45);
\node[above, yshift=2mm] at (b) {$b$};
\node[draw, circle, fill=green, inner sep=0pt, minimum width=2mm] at (7, -4.45) {};
\coordinate (c) at (7, -4.45);
\node[above, yshift=2mm] at (c) {$c$};
\node[draw, circle, fill=green, inner sep=0pt, minimum width=2mm] at (9, -6.05) {};
\coordinate (d) at (9, -6.05);
\node[above, yshift=2mm] at (d) {$d$};
\node[draw, circle, fill=green, inner sep=0pt, minimum width=2mm] at (10, -7.0) {};
\coordinate (e) at (10, -7.0);
\node[above, yshift=2mm] at (e) {$e$};
\fill[green, opacity=0.4] (a) -- (b) -- (c) -- (d) -- (e) -- cycle;
\draw[black, thick] (a) -- (b) -- (c) -- (d) -- (e) -- cycle;
	\end{tikzpicture}
	\caption{\textbf{Case 1:} all five points on the same curve necessarily create a convex $5$-gon.}
	\end{subfigure}
	\hfill
	\begin{subfigure}{0.3\textwidth}
		\begin{tikzpicture}[scale=0.3]
\node[inner sep=1pt, circle, thick, fill=blue, scale=1] at (1, 2.05) {};
\node[inner sep=1pt, circle, thick, fill=blue, scale=1] at (1, -2.05) {};
\node[inner sep=1pt, circle, thick, fill=blue, scale=1] at (2, 2.2) {};
\node[inner sep=1pt, circle, thick, fill=blue, scale=1] at (2, -2.2) {};
\node[inner sep=1pt, circle, thick, fill=blue, scale=1] at (3, 2.45) {};
\node[inner sep=1pt, circle, thick, fill=blue, scale=1] at (3, -2.45) {};
\node[inner sep=1pt, circle, thick, fill=blue, scale=1] at (4, 2.8) {};
\node[inner sep=1pt, circle, thick, fill=blue, scale=1] at (4, -2.8) {};
\node[inner sep=1pt, circle, thick, fill=blue, scale=1] at (5, 3.25) {};
\node[inner sep=1pt, circle, thick, fill=blue, scale=1] at (5, -3.25) {};
\node[inner sep=1pt, circle, thick, fill=blue, scale=1] at (6, 3.8) {};
\node[inner sep=1pt, circle, thick, fill=blue, scale=1] at (6, -3.8) {};
\node[inner sep=1pt, circle, thick, fill=blue, scale=1] at (7, 4.45) {};
\node[inner sep=1pt, circle, thick, fill=blue, scale=1] at (7, -4.45) {};
\node[inner sep=1pt, circle, thick, fill=blue, scale=1] at (8, 5.2) {};
\node[inner sep=1pt, circle, thick, fill=blue, scale=1] at (8, -5.2) {};
\node[inner sep=1pt, circle, thick, fill=blue, scale=1] at (9, 6.05) {};
\node[inner sep=1pt, circle, thick, fill=blue, scale=1] at (9, -6.05) {};
\node[inner sep=1pt, circle, thick, fill=blue, scale=1] at (10, 7.0) {};
\node[inner sep=1pt, circle, thick, fill=blue, scale=1] at (10, -7.0) {};
\node[draw, circle, fill=red, inner sep=0pt, minimum width=2mm] at (3, 2.45) {};
\coordinate (a) at (3, 2.45);
\node[above, yshift=1mm] at (a) {$b$};
\node[draw, circle, fill=red, inner sep=0pt, minimum width=2mm] at (4, 2.8) {};
\coordinate (b) at (4, 2.8);
\node[above, yshift=1mm] at (b) {$c$};
\node[draw, circle, fill=red, inner sep=0pt, minimum width=2mm] at (8, -5.2) {};
\coordinate (c) at (8, -5.2);
\node[below, yshift=-1mm] at (c) {$e$};
\node[draw, circle, fill=red, inner sep=0pt, minimum width=2mm] at (7, -4.45) {};
\coordinate (d) at (7, -4.45);
\node[below, yshift=-1mm] at (d) {$d$};
\node[draw, circle, fill=red, inner sep=0pt, minimum width=2mm] at (1, -2.05) {};
\coordinate (e) at (1, -2.05);
\node[above, yshift=-5mm] at (e) {$a$};
\fill[red, opacity=0.4] (a) -- (b) -- (c) -- (d) -- (e) -- cycle;
\draw[black, thick] (a) -- (b) -- (c) -- (d) -- (e) -- cycle;
	\end{tikzpicture}
	\caption{\textbf{Case 2:} three points on one curve and two on the other never make a convex $5$-gon.}
	\end{subfigure}
	\hfill
	\begin{subfigure}{0.3\textwidth}
	\begin{tikzpicture}[scale=0.3]
\node[inner sep=1pt, circle, thick, fill=blue, scale=1] at (1, 2.05) {};
\node[inner sep=1pt, circle, thick, fill=blue, scale=1] at (1, -2.05) {};
\node[inner sep=1pt, circle, thick, fill=blue, scale=1] at (2, 2.2) {};
\node[inner sep=1pt, circle, thick, fill=blue, scale=1] at (2, -2.2) {};
\node[inner sep=1pt, circle, thick, fill=blue, scale=1] at (3, 2.45) {};
\node[inner sep=1pt, circle, thick, fill=blue, scale=1] at (3, -2.45) {};
\node[inner sep=1pt, circle, thick, fill=blue, scale=1] at (4, 2.8) {};
\node[inner sep=1pt, circle, thick, fill=blue, scale=1] at (4, -2.8) {};
\node[inner sep=1pt, circle, thick, fill=blue, scale=1] at (5, 3.25) {};
\node[inner sep=1pt, circle, thick, fill=blue, scale=1] at (5, -3.25) {};
\node[inner sep=1pt, circle, thick, fill=blue, scale=1] at (6, 3.8) {};
\node[inner sep=1pt, circle, thick, fill=blue, scale=1] at (6, -3.8) {};
\node[inner sep=1pt, circle, thick, fill=blue, scale=1] at (7, 4.45) {};
\node[inner sep=1pt, circle, thick, fill=blue, scale=1] at (7, -4.45) {};
\node[inner sep=1pt, circle, thick, fill=blue, scale=1] at (8, 5.2) {};
\node[inner sep=1pt, circle, thick, fill=blue, scale=1] at (8, -5.2) {};
\node[inner sep=1pt, circle, thick, fill=blue, scale=1] at (9, 6.05) {};
\node[inner sep=1pt, circle, thick, fill=blue, scale=1] at (9, -6.05) {};
\node[inner sep=1pt, circle, thick, fill=blue, scale=1] at (10, 7.0) {};
\node[inner sep=1pt, circle, thick, fill=blue, scale=1] at (10, -7.0) {};
\node[draw, circle, fill=red, inner sep=0pt, minimum width=2mm] at (3, -2.45) {};
\coordinate (a) at (3, -2.45);
\node[above, yshift=2mm] at (a) {$a$};
\node[draw, circle, fill=red, inner sep=0pt, minimum width=2mm] at (8, 5.2) {};
\coordinate (b) at (8, 5.2);
\node[right, yshift=-2mm] at (b) {$b$};
\node[draw, circle, fill=red, inner sep=0pt, minimum width=2mm] at (10, -7.0) {};
\coordinate (c) at (10, -7.0);
\node[right, yshift=2mm] at (c) {$e$};
\node[draw, circle, fill=red, inner sep=0pt, minimum width=2mm] at (9, -6.05) {};
\coordinate (d) at (9, -6.05);
\node[below, yshift=-1mm, xshift=-1mm] at (d) {$d$};
\node[draw, circle, fill=red, inner sep=0pt, minimum width=2mm] at (8.1, -5.2) {};
\coordinate (e) at (8.1, -5.2);
\node[below, yshift=-1mm] at (e) {$c$};
\fill[red, opacity=0.4] (a) -- (b) -- (c) -- (d) -- (e) -- cycle;
\draw[black, thick] (a) -- (b) -- (c) -- (d) -- (e) -- cycle;
	\end{tikzpicture}
	\caption{\textbf{Case 3:} four points on one curve and a single point on the other never make a convex $5$-gon.}
	\end{subfigure}
	\caption{Illustration of the proof sketch for the parabolic construction. Curvature of $L^\top$ and $L^\bot$ has been scaled for illustrative purposes.}
	\label{fig:construction-2-proof-sketch}
\end{figure}
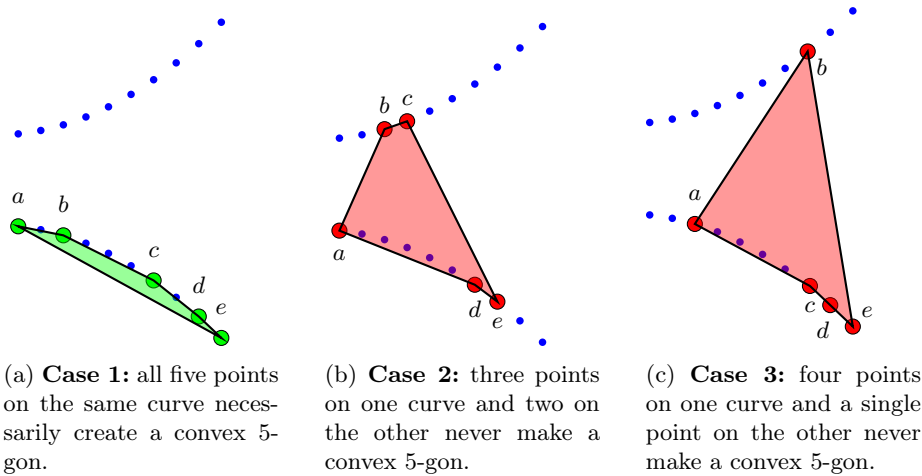
\begin{proof}[Proof sketch, illustrated in~\Cref{fig:construction-2-proof-sketch}]
    	If we consider any set of five points that are either fully contained in $L^\top := \bigcup_{i} p^\top_i$ or fully contained in $L^\bot := \bigcup_{i} p^\bot_i$, then they must define a convex $5$-gon, due to the convexity (resp.~concavity) of $L^\top$ (resp. $L^\bot$). There are exactly $\binom{\lceil n/2 \rceil}{5} + \binom{\lfloor n/2 \rfloor}{5}$ sets of $5$ points that can be chosen in this way. It remains to argue that any $5$-gon $P$ that intersects both $L^\top$ and $L^\bot$ cannot be convex. By the pigeonhole principle, there is one curve from $L \in \{L^\bot, L^\top\}$ such that $|L \cap P| \geq 3$, and for this sketch we can assume that $L = L^\bot$ without loss of generality. There are now two cases, illustrated in~\Cref{fig:construction-2-proof-sketch}: either $|L^\bot \cap P| = 3$ or $|L^\bot \cap P| = 4$, and both of them lead to non-convex $5$-gons due to the concavity of $L^\bot$.
\end{proof}

\subsection{The Parabolic Construction}
\label{subsec:parabolic}

A direct generalization of~\Cref{fig:realization-parabolic} consists on constructing, for any given $n$:
\[
	p^\top_i = \left(i, 2  + \frac{i^2}{n^2}\right), \forall i \in  \left[\left\lfloor\frac{n}{2}\right\rfloor \right] \quad \text{and} \quad
	p^\bot_i = \left(i, -2  - \frac{i^2}{n^2}\right), \forall i \in \left[\left\lceil\frac{n}{2}\right\rceil \right].\\ 
\]

An illustration of this construction is presented in~\Cref{fig:construction-2}. We remark that, albeit having a different use, the parabolic construction seems to be equivalent to the \emph{double chain} idea in~\cite{Garcia_Noy_Tejel_2000, Aichholzer_et_al_2015}, but whose application to this problem we discovered thanks to the obtained realizations.

\begin{proposition}
The number of convex $5$-gons obtained by applying the parabolic construction on $n$ points is exactly $\binom{\lfloor n/2 \rfloor}{5} + \binom{\lceil n/2 \rceil}{5}$.
\end{proposition}


As a direct consequence of the above constructions, we obtain the following upper bound for $\mu_5(n)$.
\begin{theorem}
	Let $\mu_5(n)$ denote the minimum number of convex pentagons when $n$ points are placed in the plane in general position. Then, $\mu_5(n)$ satisfies the inequality:
	\(
		\mu_5(n) \leq \binom{\left\lfloor n/2 \right\rfloor}{5} + \binom{\left\lceil n/2 \right\rceil}{5}.
	\)\label{thm:upper-bounds}
	\end{theorem}

	\begin{conjecture}
	The bounds of~\Cref{thm:upper-bounds} are tight.\label{conjecture:main}
	\end{conjecture}
	
	We remark that albeit the pinwheel construction can be deemed unnecessary as~\Cref{thm:upper-bounds} is already implied by the parabolic construction, we consider it an interesting example of how the diversity of solutions to a MaxSAT problem can translate to a diversity of constructions, or proofs, of a mathematical fact.

\section{Odd-Even Implication}
\label{sec:odd-even}
Another piece of evidence for~Conjecture~\ref{conjecture:main} is given by the following pen-and-paper theorem: if the conjecture holds for $2n-1$ points, then it must hold for $2n$ points.

\begin{proposition}
	 If for some $n > 5$ it holds that $\mu_5(2n-1) = \binom{n}{5} + \binom{n-1}{5}$, then $\mu_5(2n) = 2\binom{n}{5}$.
	\label{lemma:odd-induction}
\end{proposition}
	
In order to prove this, we will generalize a folklore idea~\cite{Brodsky_Durocher_Gethner_2001,Scheinerman_Wilf_1994}.

\begin{lemma}
	Let $m$ and $r$ be values such that $\mu_k(m) \geq r$. Then for every $n \geq m$ we have $\mu_k(n) \geq r \cdot \binom{n}{m}/\binom{n-k}{m-k}  = r \cdot \binom{n}{k}/\binom{m}{k}$.
	\label{lemma:folklore}
\end{lemma}
\begin{proof}[Proof sketch]
For each of the $\binom{n}{m}$ subsets of $m$ points, we know there will be at least $r$ convex $k$-gons. However, this will count multiple times a fixed convex $k$-gon that appears in many $m$-point subsets. In particular, each convex $k$-gon will be counted $\binom{n-k}{m-k}$ times this way, thus yielding the first inequality.  
The equality comes simply from:
\begin{align*}
\binom{n}{m}/ \binom{n-k}{m-k} &= \frac{n!}{m! \, (n-m)!} \cdot \frac{(m-k)! \, (n-m)!}{(n-k)!} \\
&= \frac{n!}{(n-k)! \, k!} \cdot \frac{(m-k)! \, k!}{m!} = \binom{n}{k}  / \binom{m}{k}. 
\end{align*}
\end{proof}
\begin{corollary}
	The limits $c_k := \lim_{n \to \infty} \mu_k(n) /{ \binom{n}{k}}$ are well defined.
\end{corollary}
\begin{proof}
	Use~\Cref{lemma:folklore} with $m = n-1$ and $r = \mu_{k}(n-1)$, to get 
	\[
		\mu_k(n) / \binom{n}{k}\geq \mu_k(n-1) / \binom{n-1}{k}, 
	\]
	which implies the sequence $\mu_k(n) / \binom{n}{k}$ is non-decreasing, and as it is clearly bounded above by $1$, we conclude.
\end{proof}

Moreover, using~\Cref{lemma:folklore} with $k=5$, $m = n-1$, and $r = \mu_{k}(n-1)$ yields:
\begin{corollary}
	For any $n > 5$, $\mu_5(n) \geq \frac{n}{n-5} \cdot \mu_5(n-1)$.
	\label{cor:induction}
\end{corollary}

We are now ready to prove~\Cref{lemma:odd-induction}.
\begin{proof}[Proof of~\Cref{lemma:odd-induction}]
	By using~\Cref{thm:upper-bounds} we have \(
		\mu_5(2n) \leq 2\binom{n}{5}.	
	\)
	Now, to argue that equality is achieved, we use~\Cref{cor:induction} to obtain that
	\begin{align*}
		\mu_5(2n) &\geq \frac{2n}{2n-5} \cdot \mu_5(2n-1) = \frac{2n}{2n-5} \left(\binom{n}{5} + \binom{n-1}{5}\right)\\
				 &= \binom{n}{5} + \frac{5}{2n-5} \binom{n}{5} + \frac{2 \cdot n! \cdot (n-5)}{(2n-5)(n-5)!\,5!}\\
				 &= \binom{n}{5} + \frac{5}{2n-5} \binom{n}{5} + \frac{2(n-5)}{2n-5} \binom{n}{5} = 2\binom{n}{5}.	
	\end{align*}
	The two inequalities imply that $\mu_5(2n) = 2\binom{n}{5}$ and conclude the proof.
\end{proof}

\section{MaxSAT Verification}\label{sec:maxsat}

The Stochastic Local Search presented in~\Cref{sec:SLS} gives us upper bounds on the value of $\mu_5(n)$ for specific values of $n$, 
on the condition that the associated signotope assignments are realizable. Thanks to~\Cref{thm:upper-bounds}, obtained through the constructions of~\Cref{sec:constructions}, we know the upper bounds found through SLS are indeed true upper bounds for $\mu_5(n)$. However, it could be the case that better bounds could be found since Stochastic Local Search does not provide any guarantees of optimality. To further support our conjecture, we will use Maximum Satisfiability (MaxSAT)~\cite{handbook-maxsat} solvers to find the \emph{optimal value} of $\mu_5(n)$ for some values of $n$. In particular, we show that $\mu_5(9) = 1$, $\mu_5(11) = 7$, $\mu_5(13) = 27$, and $\mu_5(15) = 77$. 
Concretely, as every set of points in the plane is captured by the signotope abstraction, if no assignment of signotopes induces fewer than $m$ convex $5$-gons for a given value of $n$, then indeed we conclude $\mu_5(n) \geq m$. 

\subsection{MaxSAT Encoding}

MaxSAT is an optimization variant of SAT, where, given an unsatisfiable formula, the goal is to maximize the number of satisfied clauses.
MaxSAT can be extended to include two sets of clauses: \emph{hard} and \emph{soft}. An optimal assignment for a MaxSAT problem satisfies all hard clauses while maximizing the number of satisfied soft clauses.
To build a MaxSAT encoding for this problem we first introduce \emph{relaxation} variables $\relaxp{a}{b}{c}{d}{e}$ to denote whether the 5-gon with vertices
$a$, $b$, $c$, $d$, and $e$ is convex and thus must be avoided. We then modify clauses (\ref{eq:hard1})-(\ref{eq:hard8}) by adding the literal $\relaxp{a}{b}{c}{d}{e}$ to each of them.

The axiom clauses (\ref{eq:axiom1})-(\ref{eq:axiom4}) in~\Cref{sec:background} and the modified clauses (\ref{eq:hard1})-(\ref{eq:hard8}) are defined as being hard. 
Finally we introduce the following soft clauses:
\begin{eqnarray}
\bigwedge_{(a,b,c,d,e) \in S} \big(\relaxn{a}{b}{c}{d}{e}\big).
\label{eq:soft}
\end{eqnarray}
\begin{table}[t]
        \centering
        \caption{Number of variables (\#Vars), hard clauses (\#Hard), soft clauses (\#Soft), and symmetry breaking clauses (\#Symmetry).\!\!\!\!}\label{tbl:formulas}
        \begin{tabular}{lC{1.6cm}C{1.6cm}C{1.6cm}C{1.6cm}}
        \toprule
        Instance & \#Vars & \#Hard & \#Softs & \#Symmetry\\
        \midrule
        $\mu_5(9)$ & 210	& 2\,016 	& 126  & 28\\
        $\mu_5(11)$ & 627 & 6\,336 & 462 & 45\\
        $\mu_5(13)$ & 1\,573 & 16\,016  & 1\,287 & 66\\
        $\mu_5(15)$ & 3\,458 	& 34\,944	& 3\,003	& 91\\
        \bottomrule
        \end{tabular}
        
\end{table}

The MaxSAT formulas for $\mu_5(n)$ are modest in size for small $n$, with $\mu_5(15)$ featuring about 3\,500 variables, 35\,000 hard clauses, and 3\,000 soft clauses. 
\Cref{tbl:formulas} shows the size of the MaxSAT formula for $\mu_5(n)$ (including symmetry-breaking clauses described next in~\Cref{subsec:symmetry}). 


\paragraph{\bf Symmetry Breaking.} 
\label{subsec:symmetry}
\begin{table}[t]
\small
\centering
\caption{Experimental results \textbf{without symmetry-breaking constraints}: Wall clock time in seconds to solve $\mu_5(n)$ with a time limit of 18,000 seconds (5 hours) per instance. A `$-$' denotes a timeout was reached, and optimality was not proven. For the cube-and-conquer approach (C\&C), we also include in parenthesis the sum of the CPU time needed to solve all disjoint formulas.}
\label{tab:results-symmetry}
\vspace{0.4em}
\begin{tabular}{>{\centering\arraybackslash}p{0.5cm}|l|r|r|r|r}
\cline{2-6}
&Solver & $\mu_5(9)$ & $\mu_5(11)$ & $\mu_5(13)$ & $\mu_5(15)$ \\
\hline
\multirow{4}{*}{\rotatebox[origin=c]{90}{Sequential}} & \evalmaxsat & 5.01 & 485.08 & $-$ & $-$ \\
&\maxcdcl & \textbf{0.02} & 35.28 & $-$ & $-$ \\ 
&\pacose & 0.02 & 99.69 & $-$ & $-$ \\
&\maxhs & 0.03 & $-$ & $-$ & $-$ \\
\hline
\multirow{4}{*}{\rotatebox[origin=c]{90}{C\&C}}&\evalmaxsat & 4.98 (39.40) & 287.27 (3,728.45)& $-$ & $-$\\
&\maxcdcl & 0.02 (0.14) & \textbf{2.27} (19.42) & \textbf{1,004.19} (18,133.54) & $-$ \\ 
&\pacose & 0.02 (0.14) & 5.06 (61.11) & $-$ & $-$\\
&\maxhs & 0.02 (0.17) & $-$ & $-$ & $-$ \\
\hline
\end{tabular}
\end{table}

Adding symmetry-breaking predicates that remove equivalent solutions can prune the search space and improve the performance of SAT solvers~\cite{symmetry-breaking1,symmetry-breaking2}  and is beneficial as well for MaxSAT solvers. For $\mu_5(n)$, we can break some symmetries by adding hard unit clauses $\orientp{1}{b}{c}$ with $b < c$ so that only solutions where points appear in counterclockwise order with respect to $p_1$ are obtained. 
A proof of correctness for this symmetry breaking is presented in~\cite[Lemma 1]{Scheucher_2020}.
By comparing~\Cref{tbl:formulas,tab:results,tab:results-symmetry}, we observe that adding just a few unit clauses to break symmetries has a significant impact on MaxSAT solver performance. For instance, $\mu_5(15)$ is unsolvable with any approach, compared to a 32-minute solve time with \maxcdcl (see Table~\ref{tab:results} in Section~\ref{sec:maxsat}). For $\mu_5(13)$, \maxcdcl with cube-and-conquer and symmetry-breaking predicates solves it in 7.69 seconds, compared to over 1\,000 seconds without symmetry breaking. This effect is also seen with other solvers, highlighting the importance of symmetry breaking in practical problem-solving.

\paragraph{\bf MaxSAT Approaches.}

MaxSAT solvers employ various strategies for finding optimal solutions. This paper explores four MaxSAT solvers (\evalmaxsat, \maxcdcl, \pacose, and \maxhs) that have excelled in the annual MaxSAT Evaluations.\footnote{\url{https://maxsat-evaluations.github.io/}} 
\evalmaxsat~\cite{evalmaxsat} uses an unsatisfiability-based algorithm, beginning with a linear search from the lower bound to the optimal solution. The winning version in the MaxSAT Evaluation 2023 used an integer linear programming (ILP) solver as a preprocessing step but performed better without it in our evaluation.
\maxcdcl~\cite{maxcdcl} combines clause learning with branch-and-bound and was among the top-performing solvers in the MaxSAT Evaluation 2023. We used the default version without ILP preprocessing.
\pacose~\cite{pacose} performs a linear search, iteratively improving the upper bound until it finds an optimal solution. In our evaluation, we utilized the version from the MaxSAT Evaluation 2021.
\maxhs~\cite{maxhs} employs an implicit hitting set approach, combining SAT and ILP solvers. It was the leading solver in the MaxSAT Evaluation 2021 and continues to excel in solving MaxSAT problems.

\paragraph{\bf Cube-and-conquer.} We explore a strategy inspired by the cube-and-conquer approach to parallelizing SAT formulas~\cite{cube-and-conquer}. The key idea behind cube-and-conquer is to split a formula into $2^n$ disjoint formulas by carefully choosing $n$ variables and fixing their truth values to all possible $2^n$ combinations. For instance, consider Boolean variables $x_1$ and $x_2$ that belong to a formula $\varphi$. This formula can be split into 4 disjoint formulas with the following construction $\varphi_1 = \varphi \cup (x_1 \wedge x_2)$, $\varphi_2 = \varphi \cup (\neg x_1 \wedge x_2)$, $\varphi_3 = \varphi \cup (x_1 \wedge \neg x_2)$, and $\varphi_4 = \varphi \cup (\neg x_1 \wedge \neg x_2)$. The intuition behind this idea is that it is easier to solve $\varphi_i$ than $\varphi$ and that this approach can be used to create many disjoint formulas and enable massive parallelism. In our evaluation, we selected the variables $\orientp{3}{4}{5}$, $\orientp{5}{6}{7}$, $\dots$, $\orientp{n-2}{n-1}{n}$
(3 for $\mu_5(9)$, 4 for $\mu_5(11)$, 5 for $\mu_5(13)$, and 6 for $\mu_5(15)$), because we experimentally observed that these variables split the search into balanced subspaces. 
Note that an optimal solution for $\varphi$ corresponds to the best optimal solution for \emph{all} $\varphi_i$. Even though this is just a preliminary study on using the cube-and-conquer approach to solve MaxSAT formulas, we show that even a few variables can have a significant impact on allowing us to solve harder problems.

\paragraph{\bf Finding Optimal Values for \texorpdfstring{$\mu_5(n)$}{mu\_5(n)}.}


We run \evalmaxsat, \maxcdcl, \pacose, and \maxhs with the sequential and cube-and-conquer versions on the StarExec cluster~\cite{starexec} $-$Intel(R) Xeon(R) CPU E5-2609 @ 2.40GHz$-$with a memory limit of 32 GB. All experiments were run with a time limit of 5 hours (wall-clock time) per benchmark (which is the largest time limit allowed by StarExec).
Symmetry-breaking predicates were applied to all formulas, as they are crucial for effective problem-solving (cf.~\Cref{tab:results-symmetry,tab:results}).
\Cref{tab:results} shows that MaxSAT can be used to prove the optimality of small values of $n$ for $\mu_5(n)$. Note that the best exact values for $\mu_5(n)$ prior to this work were up to $\mu_5(11)$. By using MaxSAT we can improve the best-known bounds for $\mu_5(n)$ up to $n = 16$.\footnote{While we have exact values only up to $n=15$, the odd-even implication (see~\Cref{sec:odd-even}) guarantees that the conjecture must also hold for $n=16$.}
%
Determining $\mu_5(n)$ is a challenging problem for current MaxSAT tools, and while all of the evaluated tools could solve $\mu_5(11)$, only \maxcdcl was able to solve $\mu_5(13)$ using the sequential version and $\mu_5(15)$ with the cube-and-conquer approach.
To the best of our knowledge, this is the first example of how cube-and-conquer can improve the performance of MaxSAT solvers. For instance, \maxhs can solve $\mu_5(11)$ in 28.45, seconds while it would take 15 times more wall clock time to solve it using the sequential approach. Even when considering the sum of CPU taken by all disjoint formulas, it is still beneficial to use cube-and-conquer for most cases. The cube-and-conquer approach can leverage having multiple machines available on a cluster, such as StarExec, to improve the scalability of MaxSAT tools, and allows \evalmaxsat and \pacose to solve $\mu_5(13)$ within the allocated time budget. 
Furthermore, it improved the performance of \maxcdcl to solve $\mu_5(15)$ in approximately 32 minutes. These results encourage further exploration of a cube-and-conquer approach to solve other hard combinatorial problems with MaxSAT, and open new research directions on how to automatically select splitting variables for creating disjoint subformulas in the context of MaxSAT.
\begin{table}[t]
    \small
\centering
\caption{Experimental results \textbf{with symmetry-breaking constraints}: Wall clock time in seconds to solve $\mu_5(n)$ with a time limit of 18\,000 seconds (5 hours) per instance. A `$-$' denotes a timeout was reached, and optimality was not proven. For the cube-and-conquer approach (C\&C), we also include in parenthesis the sum of the CPU time needed to solve all disjoint formulas.}
\label{tab:results}
\vspace{0.4em}
\begin{tabular}{>{\centering\arraybackslash}p{0.5cm}|l|r|r|r|r}
\cline{2-6}
&Solver & $\mu_5(9)$ & $\mu_5(11)$ & $\mu_5(13)$ & $\mu_5(15)$ \\
\hline
\multirow{4}{*}{\rotatebox[origin=c]{90}{Sequential}} & \evalmaxsat & 3.39 & 237.59 & $-$ & $-$ \\
&\maxcdcl & \textbf{0.02} & 0.49 & 150.59 & $-$ \\ 
&\pacose & 0.03 & 1.93 & $-$ & $-$ \\
&\maxhs & 0.03 & 426.97 & $-$ & $-$ \\
\hline
\multirow{4}{*}{\rotatebox[origin=c]{90}{C\&C}}&\evalmaxsat & 0.92 (5.05) & 96.65 (873.05)& 825.11 (22\,015.37) & $-$\\
&\maxcdcl & 0.03 (0.46) & \textbf{0.15} (1.46) & \textbf{7.69} (140.88) & \textbf{1\,930.40} (66\,333.04) \\ 
&\pacose & 0.03 (0.46) & 0.19 (1.73) & 136.70 (2\,647.45) & $-$\\
&\maxhs & 0.03 (0.44) & 28.45 (93.32) & $-$ & $-$ \\
\hline
\end{tabular}
\end{table}

\vspace{-0.08em}
\paragraph{\bf Certification of Results.}

%

Unlike SAT competitions, where SAT solvers provide proofs of unsatisfiability that can be independently verified, MaxSAT solvers in the MaxSAT Evaluation do not offer proofs of optimality. Consequently, there is a possibility of incorrect results.
Recently, certain MaxSAT techniques~\cite{certified-maxsat,qmaxsatpb,veritas} have emerged, capable of generating verifiable proofs of optimality using the verifier \veripb~\cite{veripb}.  
%
We used the $\veritas$~\cite{veritas} framework to generate a certified CNF formula that encodes the $\mu_5(n$) SAT problem with an additional constraint that enforces the bound to be smaller than our conjectured best value. This is similar to what \pacose does in the last step of its search algorithm when it proves that the last solution found is optimal. We can feed the resulting formula to a SAT solver and verify the unsatisfiability proof with \veripb. This approach can solve $\mu_5(9)$ and $\mu_5(11)$ and verify both results within a few seconds. Unfortunately, larger values of $n$ are beyond the reach of this approach since this construction (like \pacose) cannot solve $\mu_5(13)$ within the 5-hour time limit.

\section{Concluding Remarks and Future Work}\label{sec:conclusion}
We have proved the upper bound of~\Cref{thm:upper-bounds} in two different ways by exhibiting two different constructions, and verified through MaxSAT that this bound is tight at least up to $n = 16$.
Moreover, we have proven that the conjecture cannot fail for the first time at an even number of points.
Following the tradition of Erd\H{o}s, we offer a reward of \$500 for the first person to prove or disprove~Conjecture~\ref{conjecture:main}.
\paragraph{\bf On the Constant \texorpdfstring{$c_5$}{c\_5}.}
Let us now discuss bounds on the constant $c_5$ that our work implies. First, we note that~\Cref{thm:upper-bounds} provides an upper bound to $c_5$ as follows.
We will use the equation \(
	\lim_{n \to \infty} \binom{n}{k}/{n^k} = \frac{1}{k!}	
\), which holds for any fixed integer $k > 0$.
Now, if we consider the subsequence of even numbers $2n$, we have
\[
	c_5 = \lim_{2n \to \infty} \frac{\mu_5(2n)}{\binom{2n}{5}} \leq \lim_{2n \to \infty} \frac{2{\binom{n}{5}}}{{\binom{2n}{5}}} = \lim_{2n \to \infty} \frac{2 \cdot {n^5} \cdot {5!}}{{(2n)}^5 \cdot {5!}} = \frac{1}{16} = 0.0625.
\]

After we had written this article, it has come to our attention that this upper bound appears in the work of Goaoc et al.~\cite{goaoc_et_al:LIPIcs.SOCG.2015.300}, however, they do not provide proof.
Nonetheless, their work provides a strong lower bound of $c_5 \geq 0.0608516$.
Improving on this lower bound through SAT solving seems very challenging, as we show next.
We know that $\mu_5(16) = 112$, from where~\Cref{lemma:folklore} yields that for $n > 16$ we have
\[
	\mu_5(n) \geq 112 \cdot \frac{{\binom{n}{16}}}{{\binom{n-5}{11}}} = 112 \cdot {\binom{n}{5}} / {\binom{16}{5}} = \frac{112}{4368} \cdot {\binom{n}{5}},
\]
from where
\(
	c_5 \geq \frac{112}{4368} \approx 0.02564.
\)
Following the same method for $n = 380$ yields $c_5 \geq 0.060857$. 
That is, improving on the bound of Goaoc et al.~\cite{goaoc_et_al:LIPIcs.SOCG.2015.300} would require solving $n=380$, which is currently out of reach.


\paragraph{\bf Open Problems.}
We offer the following challenges:
\begin{enumerate}[itemsep=0.5em]
	\item \textbf{(\$500)} Prove or disprove~Conjecture~\ref{conjecture:main}.
	\item Obtain and verify the value of $\mu_5(17)$ and $\mu_5(19)$. If those values also match our conjectured bounds (182 and 378), this would contribute with an even stronger piece of evidence for our conjecture; the first $6$ odd values of $n$ would fit the degree $5$ polynomial we proposed for odd $n$, meaning that if $\mu_5(2n+1)$ were to be a fixed polynomial of degree $5$, then it must be the one we proposed.
	\item Obtain concrete bounds and design constructions for $\mu_6(\cdot)$ and $\mu_7(\cdot)$. In particular, we have checked that the parabolic construction is not optimal for $\mu_6(\cdot)$, thus suggesting that new insights will be needed. We hope our methodology based on realizations of SLS results can be helpful in this case as well.
	\item It is well known by now that symmetry breaking can provide dramatic performance advantages for SAT-solving in combinatorial problems~\cite{symmetry-breaking1,symmetry-breaking2, Subercaseaux_Heule_2023,LPAR2023:Toward_Optimal_Radio_Colorings}. A natural question is whether the $4$-fold symmetry of the pinwheel construction can be assumed without loss of generality, and if so, what would be an efficient way of taking advantage of that fact. It is worth mentioning here that the threefold symmetry of constructions for $\mu_4(\cdot)$ has been repeatedly conjectured to achieve optimality~\cite{AFMS}, and yet remains unproven.
\end{enumerate}

\paragraph{\bf Acknowledgments.}
 This work is partially supported by the U.S. National Science Foundation (NSF) under grant CCF-2108521. We thank anonymous reviewers from CICM 2024 for their feedback and suggestions.

 \bibliographystyle{splncs04}
 \bibliography{bibliography}

\end{document}